\documentclass[12pt,reqno,a4paper]{amsart}
\usepackage{amsmath,amssymb,amsfonts}
\usepackage[mathscr]{eucal}

\newtheorem{proposition}{Proposition}
\newtheorem{theorem}{Theorem}
\newtheorem{corollary}{Corollary}

\theoremstyle{definition}
\newtheorem{remark}{Remark}
\newtheorem{example}{Example}[section]
\newtheorem{definition}{Definition}

\newcommand{\cal}{\EuScript}

\renewcommand{\leq}{\leqslant}

\DeclareMathOperator{\Maps}{Maps}

\renewcommand{\L}{{\cal L}}

\newcommand{\RR}{\mathbb R}

\newcommand{\p}{\partial}

\newcommand{\g}{\mathbf{g}}
\def\a{\alpha}

\def\d{\delta}

\def\s{\sigma}

\def\G {\Gamma}
\def\X  {{\bf X}}

\def\N {{\cal N}}
\def\Y  {{\bf Y}}

\def\A  {{\bf A}}
\def\B  {{\bf B}}
\def\F {{\cal F}}

\def\S {{\mathbf S}}
\def\E {{\mathbf E}}
\def\e {{\mathbf e}}
\def\P {{\mathbf P}}

\def\l {{\lambda}}
\def\g {{\gamma}}

\def\wD #1 {{\widehat {\Delta_{#1}}}}

\def\w {\omega}
\def\om {{\boldsymbol \omega}}

\newcommand{\bs}{{\boldsymbol{s}}}

\newcommand{\rh}{{\boldsymbol{\rho}}}
\newcommand{\io}{{\boldsymbol{\iota}}}
\newcommand{\et}{{\boldsymbol{\eta}}}

\title[Odd Laplacians:  potential and modular class]
{Odd Laplacians: geometrical meaning of potential, and modular class}

\author{H.~M.~Khudaverdian}

\author{M.~Peddie}

\address{School of Mathematics,  University of Manchester,
 Oxford Road,  Manchester   M13 9PL,  UK}
\email{khudian@manchester.ac.uk\\matthew.peddie@manchester.ac.uk }

\keywords{operators on half-densities, odd Poisson manifold,
modular class,  odd symplectic manifold, canonical odd Laplacian,
compensating field.}

\subjclass[2000]{53D17, 58A50, 81R99}

\begin{document}

\maketitle


\begin{abstract}
   A second order self-adjoint operator $\Delta=S\partial^2+U$ is
uniquely defined by its
principal symbol $S$ and potential $U$ if it acts on half-densities.
We analyse the potential $U$ as  a compensating field (gauge field)
 in the sense that it compensates
the action of coordinate transformations  on the second
derivatives in the same way as an affine connection
compensates the action of coordinate transformations
on first derivatives in the first order operator,
a covariant derivative, $\nabla=\partial+\Gamma$.
   Usually a potential $U$ is derived from other geometrical constructions
such as a volume form, an affine connection, or a Riemannian
structure, etc.
  The story is different if $\Delta$ is an odd operator on a supermanifold.
In this case the second order potential becomes a primary object.
 For example, in the case of an odd symplectic supermanifold,
the compensating field of the canonical odd Laplacian
    depends only on this
  symplectic structure, and can be expressed
 by the formula obtained by K.Bering.
   We also study modular classes of odd Poisson manifolds
   via $\Delta$-operators, and consider an example of a
  non-trivial modular class
    which is related with the Nijenhuis bracket.

\end{abstract}

 \section {Second order self-adjoint operator on half-densities}

  Let $\S=S^{ab}\p_b\otimes\p_a$ be a rank 2,
contravariant
symmetric tensor field
 on a (super)manifold $M$.
 We consider the class  $\F_\S$ of second
order operators acting on half-densities on $M$
 such that every operator
  $\Delta$ of this class  obeys the following conditions:
\begin{itemize}

\item
  $\Delta$ has principal symbol $\S$; in local coordinates
  $\Delta={1\over 2}S^{ab}(x)\p_b\p_a+\dots$;
\item
  Operator $\Delta$ is self-adjoint:
$\Delta^*=\Delta$, i.e.
\end{itemize}
             \begin{equation}\label{defofadjoint1}
\langle\Delta\bs_1,\bs_2\rangle=
\langle\bs_1,\Delta^*\bs_2\rangle\,,
\end{equation}
    where $\langle\,,\,\rangle$ is the natural scalar product on
half-densities:  
         \begin{equation*}\label{scalarproduct}
\langle\bs_1,\bs_2\rangle=\int_M s_1(x)s_2(x)|Dx|\,,
\end{equation*}
$\bs_1=s_1(x)\sqrt {|Dx|}$,
   $\bs_2=s_2(x)\sqrt {|Dx|}$ are two arbitrary
half-densities (with compact
support).

\begin{remark}
A density of weight $\lambda$, $\bs=s(x)|Dx|^\l$,
is multiplied by the $\l^{th}$ power
of the Jacobian under a change of coordinates. Densities of weight $\l=0$ are
functions on the manifold, and those of weight $\l=1$ are volume forms. In this article
we mostly consider densities of weight $\l={1\over 2}$, half-densities.
\end{remark}

\begin{proposition}
 For an arbitrary rank 2, contravariant symmetric tensor field
$\S=\S(x)$
on a manifold $M$, the class  $\F_\S$ is not empty, and any
two operators in this class  differ by a scalar function:
 if $\Delta,\Delta'\in \F_\S$ then
                \begin{equation}\label{difisscalarfunction}
  \Delta' - \Delta = F(x)\,.
                \end{equation}
 One can say that $\F_\S$ is an affine space of second order operators
associated with the vector space of functions on $M$.

\end{proposition}
\begin{proof}  First prove \eqref{difisscalarfunction}.
Indeed, operator $\Delta'-\Delta$ has to be an operator of order $\leq 1$
since both operators $\Delta$ and $\Delta'$ have the same
principal symbol. The self-adjointness  condition
implies that the order of this operator
is not equal to $1$; if
  $\Delta'-\Delta=L^a(x)\p_a+\dots$ then self-adjointness implies
              \begin{equation*}\label{firstorderantiselfadjoint}
      L^a(x)\p_a+\dots=
\left(L^a(x)\p_a+\dots\right)^*=-L^a(x)\p_a+\dots \Rightarrow L^a\equiv 0\,.
               \end{equation*}
Thus we come to equation \eqref{difisscalarfunction}.

   To prove that the class $\F_\S$
is not empty consider (using a partition of unity
argument,) an arbitrary volume form $\rh=\rho(x)|Dx|$.
Now using this volume
form we construct a second order operator belonging to the class
   $\F_\S$.
  Every half-density $\bs=s(x)\sqrt{|Dx|}$ defines a function
  $f_\bs={\bs\over\sqrt \rh}={s(x)\over\sqrt{\rho(x)}}$.
This function defines a
  covector field
$df_\bs$,
and this covector field  defines a
vector field $\X_\bs=\S df_\bs$ via the contravariant
tensor field $\S$.
   Considering the divergence of this
vector field with respect to the volume form $\rh$, we come
to the operator $\Delta^{(\rh)}$ on half-densities
defined in the following way:
          \begin{equation*}\label{deltaviarho}
      \Delta^{(\rh)} \bs=
              {1\over 2}\sqrt\rh
         {\rm div\,}_\rh \X_\bs=
        {1\over 2}\sqrt \rh {\rm div\,}_\rh \left(\S df_\bs\right)=
               {1\over 2}\sqrt\rho(x)
               {1\over \rho(x)}
                      \p_a
                    \left(
         \rho (x)S^{ab}(x)\p_b
               \left(
         {s(x)\over \sqrt\rho(x)}
           \right)\right)\sqrt{|Dx|}
           \end{equation*}
            \begin{equation}\label{deltaviarho2}
       =  {1\over 2} \left(
           \p_a\left(S^{ab}{\p_b s}\right)
       -{1\over 2}
             \p_a
           \left(
               S^{ab}
       {\p_b\log\rho}\right)s-
       {1\over 4}{\p_a\log\rho}S^{ab}
       {\p_b\log\rho}s
          \right)\sqrt{|Dx|}.
            \end{equation}
It is easy to check that it  is a
self-adjoint operator. Every operator $\Delta\in \F_\S$
differs from $\Delta^{(\rh)}$
by a scalar function:
        $
\F_\S\ni\Delta=\Delta^{(\rh)}+F(x).
        $
\end{proof}

We see that self-adjointness uniquely defines a second order operator
on half densities by its symbol up to a function.
 It is useful to look at the analogous statement for
  first order operators.

\begin{proposition}\label{firstorderisliederivative}
Let $L$ be first order anti-self-adjoint operator on half-densities:
   $L^*=-L$ with principal symbol, vector field  $\X=X^a\p_a$.
Then operator $L$ is the Lie derivative of half-densities along the
field $\X$, $L=\L_\X$,
                   \begin{equation}\label{liederivativeofhalfdensity}
\L_\X\colon\quad \L_\X\bs=\L_\X(s(x)\sqrt{|Dx|})=
         \left(
  X^a(x)\p_a s(x)+{1\over 2}\p_a X^a(x) s(x)
              \right)
         \sqrt{|Dx|}\,.
           \end{equation}

\end{proposition}

\begin{proof}  Let $L=X^a\p_a+B$.
Condition of anti-self-adjointness means that
         $$
  L^*=\left(X^a\p_a+B\right)^*=-X^a\p_a-\p_a X^a+B=
-\left(X^a\p_a+B\right)\Rightarrow B={1\over 2}\p_aX^a\,,{\rm and}\,
   L=\L_\X\,.
      $$
\end{proof}

\begin{corollary}\label{commutator}
  Let $\Delta$ be an arbitrary self-adjoint second order operator
on half-densities with principal symbol $\S=S^{ab}\p_b\otimes \p_a$,
 $\Delta\in\F_\S$.  Then
 for an arbitrary function $f$,
    \begin{equation}\label{formulafromlmp}
   \Delta(f\bs)=f\Delta\bs+\L_{D_f},\quad
   \hbox{where}\, D_f=\S df=S^{ab}\p_bf\p_a\,.
     \end{equation}
\end{corollary}
\begin{proof}
The operator $[\Delta, f]=\Delta\circ f-f\circ\Delta$
is a first order operator with principal symbol,
vector field  $\S df$. It is an anti-self-adjoint operator:
  $[\Delta,f]^*=-[\Delta,f]$.
Proposition \ref{firstorderisliederivative}
then implies that
equation \eqref{formulafromlmp} is obeyed.
\end{proof}

\begin{remark}\label{defwithoutvolumeform}
  We proved that the class $\F_\S$ is not empty by constructing an operator
in this class via a volume form. For further considerations,
it is useful to present a construction where one can
define an operator in the class $\F_\S$ without using a volume form.
  To do this we note that for any principal symbol $\S(x)$
one can consider a finite set of vector fields $\{\X_\l,\Y_\l\}$
such that
$\S$ can be decomposed as
                 $$
  \S(x)=\sum_\l\X_\l(x)\otimes_{_{\rm sym}}\Y_\l(x)\,,\quad
         S^{ab}(x)={1\over 2}
    \sum_\l\left(X^a_\l(x)Y^b_\l(x)+Y^a_\l(x)X^b_\l(x)\right)\,.
                 $$
Then consider an operator
                \begin{equation}\label{arbitrarydelta}
\Delta_\S=\Delta={1\over 2}
    \sum_\l\left(
         \L_{\X_\l}
         \L_{\Y_\l}+
         \L_{\Y_\l}
         \L_{\X_\l}
          \right)\,,
                 \end{equation}
where $\L_\X$ is the
Lie derivative \eqref{liederivativeofhalfdensity}
of half-densities along the vector
field $\X$. The principal symbol of this operator is equal to $\S(x)$.
 Due to the anti-self-adjointness of the Lie derivative of half-densities,
this operator is self-adjoint.
Hence equation \eqref{arbitrarydelta} defines a
self-adjoint operator on half-densities with principal symbol $\S$.

 The relation between $\S$ and the set
$\{\X_\l,\Y_\l\}$ of vector fields is not canonical,
Considering different decompositions we come to different operators in
the class $\F_\S$.

\end{remark}

One can see that an arbitrary $\Delta\in \F_\S$ has the following
 appearance in local coordinates,
     \begin{equation}\label{firstappearanceofpotential}
   \Delta={1\over 2}\left(S^{ab}(x)\p_b\p_a+
    \p_b S^{ba}(x)\p_a+U(x)\right)\,.
   \end{equation}
\begin{definition}
We say that $U(x)$ is a {\it second order
compensating field} or just a {\it compensating field}
of the second order
operator $\Delta\in\F_\S$.
\end{definition}

For example, for the operator \eqref{deltaviarho2}
the compensating  field is equal to
         \begin{equation}\label{potentialforstandardoperator}
        U^{(\rh)}(x)=-{1\over 2}\p_a\left(S^{ab}(x)
       \p_b\log\rho(x)\right)-
       {1\over 4}\p_a\log\rho(x)S^{ab}(x)
       \p_b\log\rho(x)\,.
              \end{equation}

Let us see how the compensating field transforms
under a change of local coordinates $(x^a)\mapsto (x^{a'})$.
  The expression for a half-density in new coordinates is
                  $
\bs= s'(x')\sqrt{|Dx'|}=s'(x'(x))\sqrt{
                 J_{\{x',x\}}}
                        \sqrt {|Dx|}
                  $, thus $s(x)=s'(x'(x))\sqrt{
                 J_{\{x',x\}}}$,
where  we denote
          \begin{equation}\label {notationforcolumechangings}
          J_{\{x',x\}}=
   \det
       \left({\p x'\over \p x}\right)\,.
          \end{equation}
  Performing straightforward calculations, we
see that in the new coordinates $(x^{a'})$,
                   $$
\Delta \bs=\left(
  \p_b\left(S^{ba}(x)\p_a \left(
          s'(x')\sqrt{J_{\{x',x\}}}
                \right)\right)
+                 U(x)s'(x')\sqrt{J_{\{x',x\}}}\,\,
         \right)\sqrt{J_{\{x,x'\}}}\sqrt{|Dx'|}=
                   $$
                $$
\left(\p_{b'}\left(S^{b'a'}(x')\p_{a'} s'(x')\right)+
   U'(x')s'(x')\right)\sqrt{|Dx'|}\,,
                $$
where for tensor field $\S=S^{ab}\p_b\otimes\p_a$,
          $$
    S^{a'b'}(x')=
     {\p x^{a'}\over \p x^a}S^{ab}(x)
         {\p x^{b'}\over \p x^b}\,,
             $$
and
         \begin{equation}\label{transformofcompensationfield}
   U'(x') =U(x)+
         {1\over 2}
     \p_{a'}\left(S^{a'b'}\p_{b'}\log J\right) -
        {1\over 4}\p_{a'}\log J S^{a'b'}\p_{b'} \log J\,,
        \end{equation}
where $J=J_{\{x',x\}}$ is the Jacobian defined by equation
\eqref{notationforcolumechangings}.
\begin{remark}
In these calculations we use the fact that for
Jacobian $J=J_{\{x',x\}}$,
         $$
{\p\over \p x^r}\log J=
    {\p x^c\over \p x^{a'}}{\p^2 x^{a'}\over \p x^c\p x^r}\,,
       $$
since for an arbitrary matrix-valued function $M$, $\delta\log M=
{\rm Tr}\left(M^{-1}\delta M\right)$.
\end{remark}

\begin{proposition}\label{prop2}
   An arbitrary, self-adjoint second order differential operator
on half-densities is well-defined by two geometrical objects:

\begin{itemize}
\item

The contravariant rank $2$ symmetric tensor field $\S$
which is the
principal symbol of the operator;

\item A second order compensating field  $U(x)$, a
 geometrical object which transforms under a change of coordinates
 according to equation \eqref{transformofcompensationfield}.
\end{itemize}

  For a given tensor field $\S$, a space of compensating fields
   $U$ is an affine space associated with the vector space of functions
on $M$.

 \end{proposition}

 Consider examples of second order compensating fields.

\begin{example}  If a manifold $M$ is provided with a volume form
  $\rh=\rho(x)|Dx|$ then one can consider the second order
compensating field $U=U^{(\rh)}$ (see equations \eqref{deltaviarho2}
and \eqref{potentialforstandardoperator}). Volume form
  $\rh$ can be derived from a Riemannian metric $G$ ($\rh=\sqrt {\det G}|Dx|$),
 or a symplectic structure ($\rh=|Dx|$ in Darboux coordinates).
  See also further examples.
\end{example}

\begin{example}
A compensating field $U$ can be derived from a connection $\nabla$
on densities. If $\bs=s(x)|Dx|^\l$ is a density of weight $\l$
and $\X$ is a vector field, then the covariant derivative
of density $\bs$
along vector field $\X$ is equal to

     \begin{equation*}
\nabla_\X\bs=\nabla_{X^a\p_a}\left(s(x)|Dx|^\l\right)=
   X^a\left(\p_a s(x)+\l\g_as(x)\right)|Dx|^\l\,,
                 \end{equation*}
where $(\gamma_a)$ are such that for coordinate volume form
 $|Dx|$, $\nabla_{\p_a}|Dx|=\gamma_a|Dx|$.
Under a change of local coordinates
 $\gamma_a$ transforms in the following way:
           $$
   \gamma_a=x^{a'}_a\left(\gamma_{a'}+
   \p_{a'}\log\det \left({\p x\over \p x'}\right)\right)=
   x_a^{a'}\gamma_{a'}-x^b_{b'}x^{b'}_{ba}\,.
           $$
We use short notations for derivatives:
 $x^{a'}_a={\p x^{a'}(x)\over \p x^a}$,
 $x^{a'}_{bc}={\p x^{a'}(x)\over \p x^b \p x^c}$.
Symbols $(\gamma_a)$ can
be considered as  first order compensating fields.

  Notice that every volume form $\rh=\rho(x)|Dx|$
 induces a connection $\nabla^{(\rh)}$
on densities:
          \begin{equation*}\label{connectionviavolumeform}
        \nabla^{(\rh)}_\X (\bs)=\rh^{\l}
         \p_\X\left(\rh^{-\l}\bs\right)=
                  X^a\left(\p_a s(x)+\l\g_a^{(\rh)}\right)|Dx|^\l\,,\quad
                  \g_a^{(\rh)}=-\p_a\log\rho(x)\,.
                \end{equation*}
This is a flat connection.
  One can see that the potential $U=U^{(\rh)}$ corresponding to volume
form $\rh$ (see equation \eqref{potentialforstandardoperator}
and the previous example) can be expressed via the connection
$\nabla^{(\rh)}$ in the following way: if $\gamma_a=\gamma_a^{(\rh)}$
then
       \begin{equation}\label{potentialforstandardoperator2}
        U^{(\rh)}={1\over 2}\p_a\g^a-{1\over 4}\gamma_a\gamma^a\,,
   \quad {\rm where\,\,} \,\,\g^a=S^{ab}\g_b\,.
              \end{equation}
This
 formula works not only for a flat connection $\nabla^{(\rh)}$
defined by a volume form, but for an arbitrary connection $\nabla$
on densities. (It can be viewed in  a more general framework
if we  consider an upper connection $\gamma^a$.)

 A connection $\nabla$ on densities can be naturally produced via
an affine connection: $\g_a=-\G^b_{ba}$, where $\G^a_{bc}$
are the Christoffel symbols of the affine connection.
 If the affine connection is the Levi-Civita connection
of a Riemannian manifold $(M,G)$, then it
produces a flat connection corresponding to the canonical  volume form
of the Riemannian manifold:
        $$
\gamma_a=-\Gamma^b_{ba}=-\p_a\log\sqrt{\det G}=\g_a^{(\rh)}\,,
   {\rm where\,}\rh=\sqrt {\det G}|Dx|\,.
          $$

(See for detail \cite{KhVor2}, \cite{KhVor4}.)

\end{example}

In all previous examples the second order compensating field
arises as a secondary object from some structure, such as  a metric,
volume form or first order connection.
  In the next example we will consider a
compensating field (a potential) which arises as a  primary object.

\begin{example}\label{superspaceforcanonicalvolumeform}

For an arbitrary manifold $M$ consider its tangent bundle
 $TM$, and the supermanifold $\Pi TM$, where $\Pi$ is the
parity reversing functor in fibres. For local coordinates
 $(x^i)$ on $M$ one can assign coordinates
  $(x^i,dx^j)$ on $\Pi TM$, where $(dx^i)$ are odd coordinates,
$p(dx^i)=1$.
Functions $F(x,dx)$
on the supermanifold $\Pi TM$ can be identified with differential forms
on $M$.

 One can consider on $\Pi TM$ the canonical volume form $\rh$
which, in coordinates $(x^i,dx^j)$, is equal to the coordinate
volume form
           \begin{equation}\label{canonicalvolumeform}
          \rh=|D(x,dx)|.
          \end{equation}
This form remains invariant under a change of coordinates, since the
Berezinian (superdeterminant)
of the coordinate transformation $(x^i,dx^j)\mapsto (x^{i'}, dx^{j'})$
is equal to $1$:
  $x^{i'}=x^{i'}(x), dx^{j'}=dx^j{\p x^{j'}\over \p x^j}$ then
           \begin{equation*}
D(x',dx')=D(x,dx){\rm Ber}
      \left(
    {\p\left(x',dx'\right)\over \p\left(x,dx\right)}
       \right)=
         D(x,dx)
            {\rm Ber\,}
           \begin{pmatrix}
          {\p x^{i'}\over \p x^i} & 0\cr
         dx^j{\p^2 x^i\over \p x^i\p x^j}&{\p x^{i'}\over \p x^i}
          \end{pmatrix}=D(x,dx)\,.
          \end{equation*}
We see from equations \eqref{canonicalvolumeform}
and \eqref{potentialforstandardoperator}
that for every principal symbol $\S$ on $\Pi TM$
one can consider a canonical potential which vanishes in coordinates
  $(x^i,dx^j)$.

The constructions of this example are valid
if  $M$ is not only a  usual manifold,
but an arbitrary supermanifold. In this case the coordinates
$(x^i)$ are even and odd coordinates on $M$, and the
 parities of coordinates  $(dx^i)$ are opposite to those
of the coordinates $(x^i)$ on $M$, $p(dx^i)=p(x^i)+1$.
When there are odd variables present, functions on $\Pi TM$
correspond to so called pseudodifferential forms.

\end{example}
\begin{remark}\label{projective}{\it
Projective connection and compensating field.}

 In this article  we consider only
 operators of weight $0$, i.e. operators
 which do not change the weight of densities.
 One can consider second order operators
 of weight $\delta\not=0$.
 A contravariant symmetric tensor density
  $\S |Dx|^\delta=S^{ab}\p_a\otimes\p_b|Dx|^\delta$
defines the class $\F_\S^{(\delta)}$
of self-adjoint second order operators
 of weight $\delta$
which act on densities of the weight
   ${1-\delta\over 2}$ and takes values
 in densities of weight
   ${1-\delta\over 2}+\delta={1+\delta\over 2}$.
(See for detail \cite{KhVor2} and \cite{KhVor4}.)

In particular for the $1$-dimensional case
we come to the following very important construction.
On the line $\RR$, a canonical  symmetric tensor
density of weight  $\delta=2$, $\p_x\otimes\p_x|Dx|^2$
defines a second-order self-adjoint operator
on densities of weight ${1-\d\over 2}=-{1\over 2}$
                 \begin{equation}
\Delta=\left(\p_x^2+U(x)\right)|Dx|^2.
                  \end{equation}
  The compensating field $U(x)|Dx|^2$ is called a
  {\it projective connection}.
Under a change of coordinates $y=y(x)$
it transforms by the Schwartzian of the coordinate transformation:
    \begin{equation*}
  U'(y)|Dy|^2=U(x)|Dx|^2-{1\over 2}
                 \left(
      {y_{xxx}\over y_x}-{3\over 2}{y^2_{xx}\over y_x^2}
                 \right)
                 |Dx|^2\,.
    \end{equation*}
(See for example \cite{HitchSeg}, \cite{OvsTab} and \cite{KhVor4}.)

\end{remark}

\section {Second order operators on odd Poisson supermanifolds.}

  The considerations of the previous section can be easily
generalised for operators acting on supermanifolds,
this simply involves inserting the relevant signs
and using the super analogues for integrals, in particular
the determinant will change to the Berezinian (superdeterminant)
(as in example \eqref{superspaceforcanonicalvolumeform}).

   Most of the constructions
 in this section become meaningful
only in the super-setting.
So we will be careful about the signs arising during the calculations.
For example, the definition \eqref{defofadjoint1} of an adjoint
operator has to be rewritten
   $\langle\Delta\bs_1,\bs_2\rangle=
(-1)^{p(\Delta)p(\bs_1)}\langle\bs_1,\Delta^*\bs_2\rangle$,
where $p(\Delta),p(\bs_1)$ are  the parities of the operator $\Delta$
and $\bs_1$. For references on supermathematics
see the book of F.A.Berezin \cite{Berezin}.

   Let $M$ be a supermanifold provided with
 an odd rank $2$ contravariant symmetric tensor  
field $\E=E^{ab}(x)\p_b\otimes \p_a$.
  If $(x^a)$ are local coordinates on the manifold $M$
(hence forward we will usually suppress the prefix super,)
 then
         \begin{equation}\label{supersymmetric}
  \E\colon\quad   p(E^{ab}(x))=1+p(a)+p(b)\,,\qquad
              E^{ab}(x)=(-1)^{p(a)p(b)}E^{ba}(x)\,,
         \end{equation}
where $p(a)$ is the  parity of the coordinate $x^a$.
We will reserve the notation $\E=E^{ab}\p_b\otimes\p_a$
for such an odd rank $2$ tensor field.

Consider an arbitrary self-adjoint second order odd operator
on half-densities
with principal symbol $\E$,
            $$
\Delta\in\F_\E,\quad
\Delta={1\over 2}E^{ab}(x)\p_b\p_a+\ldots\,,
 \qquad p\left(\Delta\bs\right)=1+p(\bs)\,.
            $$
One can see that in local coordinates it has the appearance
        \begin{equation}\label{deltainlocalcoordinates}
\Delta={1\over2}\left(E^{ab}\p_b\p_a+\p_bE^{ba}\p_a+U(x)\right)\,
        \end{equation}
where second order compensating field $U(x)$ is an odd valued function.

Consider the `classical limit' of operator $\Delta$.

Let $(\,,\,)$ be the canonical even Poisson bracket
on the cotangent bundle $T^*M$:
    \begin{equation}\label{canoncialevenbracket}
   (F,G)=(-1)^{p(a)(p(F)+1)}\left({\p F\over \p p_a}{\p  G
 \over \p x^a}-(-1)^{p(a)}{\p F
  \over \p x^a}{\p  G\over \p p_a}\right)\,,
    \end{equation}
where $(x^a,p_b)$ are coordinates on $T^*M$ adjusted to local coordinates
 $(x^a)$ (under a change of local coordinates $(x^a)\mapsto (x^{a'})$,
     $p_{b'}={\p x^b\over \p x^{b'}}p_b$).
Then the  principal symbol $\E$
 defines  on functions on the manifold $M$
  the following odd bracket
                  \begin{equation}\label{derivedbracket1}
[F,G]=(-1)^{p(F)}\left(\left(H_E,F\right),G\right)=
(-1)^{p(a)p(F)}\p_a F E^{ab}\p_b G\,,\quad
 [x^a,x^b]=(-1)^{p(a)}E^{ab}\,,
                  \end{equation}
where $H_E={1\over 2}E^{ab}p_bp_a$
is an odd  Hamiltonian on $T^*M$, quadratic on the fibres of $T^*M$
(see for detail \cite{KhVor1}). This bracket is antisymmetric with respect
 to a shifted parity:
                   \begin{equation*}
   [F,G]=-(-1)^{(p(F)+1)(p(G)+1)}[G,F]\,.
                     \end{equation*}
  We say that  $(M,\E)$ is an odd Poisson manifold
if  $\E$ defines an odd Poisson bracket on $M$.
That is, if the Jacobi identity for the
bracket $[\,,\,]$ \eqref{derivedbracket1} is obeyed
            \begin{equation*}\label{jacobi}
   (-1)^{p(F)p(H)+p(G)}\left[F,[G,H]\right]+
(-1)^{p(F)p(G)+p(H)}\left[G,[H,F]\right]  
   +(-1)^{p(G)p(H)+p(F)}\left[H,[F,G]\right]=0,
             \end{equation*}
in local coordinates
         \begin{equation}\label{jacobi1}
(-1)^{p(a)(p(c)+1)}E^{ap}\p_p E^{bc}+\hbox{cyclic permutations}=0\,.
        \end{equation}
The left hand side of this equation is nothing but
the canonical bracket of the odd Hamiltonian $H_\E$ with itself.
 Jacobi identities \eqref{jacobi1} are
     obeyed if and only if $(H_E,H_E)=0$.
  (See for detail \cite{KhVor1}.)

 The fact that the tensor field $\E$ equips $M$ with
  the structure of an odd Poisson manifold
   can be clearly viewed in terms of operator $\Delta$.

  \begin{proposition}\label{deltasquare}
  Let $\E=E^{ab}\p_b\otimes \p_a$ be an odd contravariant symmetric tensor
field on a supermanifold $M$, and $\Delta$ be an arbitrary second
 order odd self-adjoint operator with  principal symbol $\E$,
  $\Delta=E^{ab}\p_b\p_a+\dots,$ ($\Delta \in \F_\E(M)$). Then
\begin{itemize}

\item    Operator $\Delta^2={1\over 2}[\Delta,\Delta]$
  is an even anti-self-adjoint operator, $(\Delta^2)^*=-\Delta^2$,
 which has order equal to $3$, $1$, or else $\Delta^2=0$.

\item   Operator $\Delta^2$ is of order $1$, or else $\Delta^2=0$,
  if and only if $\E$ defines an odd Poisson structure on $M$, i.e.
if $(M,\E)$ is an odd Poisson supermanifold.

\item If $(M,\E)$ is an odd Poisson supermanifold then
  an even anti-self-adjoint operator $\Delta^2$ is
equal to the Lie derivative along some vector field $\X$:
           \begin{equation}\label{modularvector}
                 \Delta^2=\L_\X.
           \end{equation}
This equation assigns an even vector field
  $\X=\X(\Delta)$ to the operator $\Delta$ defining the
  odd Poisson structure. We call the vector field $\X$
  the modular vector field of the operator $\Delta$.
The modular vector field $\X(\Delta)$ is a Poisson vector field,
   that is, it preserves the odd Poisson structure.

\item If $\Delta'=\Delta+F$ is another odd operator defining
the same Poisson structure $(\Delta,\Delta'\in \F_\E(M),)$
then
  \begin{equation}\label{classisthesame}
\X(\Delta')=\X(\Delta)+D_F\,,
           \end{equation}
where $D_F$ is the even Hamiltonian vector field corresponding to
the odd function $F$: $D_F G=[F,G]$ for an arbitrary function $G$.

\end{itemize}
  \end{proposition}

\begin{proof}  The first statement of the Proposition
follows from the fact that for two arbitrary
operators $\Delta_1,\Delta_2$
            $$
 [\Delta_1,\Delta_2]^*=
        [\Delta_1\circ\Delta_2-
 (-1)^{p(\Delta_1)p(\Delta_2)}\Delta_2\circ\Delta_1]^*=
        (-1)^{p(\Delta_1)p(\Delta_2)}\Delta^*_2\circ\Delta^*_1-
                   \Delta^*_1\circ\Delta^*_2=
        -[\Delta_1^*,\Delta_2^*]\,.
            $$
Thus for self-adjoint operators their commutator is anti-self-adjoint.
   Since $\Delta$ is odd, $\Delta^2={1\over 2}[\Delta,\Delta]$
is a commutator operator, so its order is less
than or equal to $3$. For the even operator $\Delta^2$,
                $$
\Delta^2=P^{abc}\p_c\p_b\p_a+\dots=
(E^{ab}\p_b\p_a+\dots)^2=2(-1)^b E^{ba}\p_aE^{cd}\p_d\p_c\p_b+\dots\,.
                 $$
The principal symbol of $\Delta^2$ defines a cubic Hamiltonian
$(H_\E,H_\E) = 2(-1)^b E^{ba}\p_aE^{cd}p_dp_cp_b$.
The vanishing of this cubic Hamiltonian
 is equivalent to the
 Jacobi identity \eqref{jacobi1} for the odd
   bracket \eqref{derivedbracket1}.
Therefore the Jacobi identity
  is obeyed if and only if
 the order of $\Delta^2$  is less than $3$. One can say more.
The order of an anti-self-adjoint operator $\Delta^2$
cannot be equal to $2$
or $0$.
 Indeed, if $\Delta^2=L^{ab}\p_b\p_a+P^a\p_a+F$ then
$\left(\Delta^2\right)^*=-\Delta^2=-L^{ab}\p_b\p_a + \dots$,
hence $L^{ab}\equiv 0$, i.e. the order of operator $\Delta^2$ is $\leq 1$.
If the $P^a$ also vanish, i.e. the order of $\Delta^2$ is $<1$,
 then anti-self-adjointness implies that $F$
also vanishes, i.e. $\Delta^2=0$.
  We see that the order of the operator $\Delta^2$ can be equal
to either $3$, $1$, or $\Delta^2=0$.
  Therefore if $\E$ defines an odd Poisson structure on $M$
then $\Delta^2$ is an even first order operator (or it vanishes).
This operator
is an anti-self-adjoint, even operator. Hence it is equal to
the Lie derivative $\L_\X$ along a vector field $\X$,
which is its principal symbol (see Proposition\,
\ref{firstorderisliederivative}):
      $\Delta^2=\L_\X=X^a\p_a+{1\over 2}(-1)^a\p_aX^a$.

 The vector field  $\X$
  preserves the operator $\Delta$:
$\L_\X \circ\Delta= \Delta^2\circ \Delta=
\Delta\circ \Delta^2=
\Delta\circ\L_\X$.
Hence it preserves the odd Poisson bracket
\eqref{derivedbracket1}, defined by the
principal symbol of this operator.
  The proof of equation \eqref{classisthesame} follows from
  Corollary\,\ref{commutator}. Since $F$ is an odd function then
            $$
(\Delta')^2=(\Delta+F)^2=
 \Delta^2+[\Delta,F]=\L_\X+\L_{\E df}=\L_{\X+D_F}\,.
            $$

\end{proof}

 Hamiltonian vector fields evidently preserve the
Poisson bracket, they are Poisson vector fields.
Consider the first Lichnerowicz-Poisson cohomology group of the
Poisson manifold $(M,\E)$
               $$
   H^1_{LP}(M,\E)={\hbox{Poisson vector fields on $M$}\over
                 \hbox{Hamiltonian vector fields on $M$}}.
               $$
Using the Proposition above we come to
\begin{definition}\label{definitionofmodularclass}
  Let $(M,\E)$ be a supermanifold provided
with an odd Poisson structure, and let $\Delta$
be an arbitrary odd second order self-adjoint
operator on half-densities, such that
$\E$ is its principal symbol, $\Delta\in \F_\E(M)$.
Let $\X=\X(\Delta)$ be the modular vector field\,\eqref{modularvector}
 of operator $\Delta$.
 Due to equation \eqref{classisthesame},
the  equivalence class  $[\X]$ in
 $H^1_{LP}(M,\E)$
 does not depend
on a choice of operator $\Delta$ in the class
$\F_\E(M)$.  We call this class the
{\it modular class of the odd Poisson supermanifold $(M,\E)$}.

\end{definition}

  It is useful to write down explicit formulae
for the modular vector field of a self-adjoint operator $\Delta$.
If $\Delta$ has the appearance
\eqref{deltainlocalcoordinates} in local coordinates, then
    \begin{equation*}\label{explicitformula}
\Delta^2=\L_\X=X^a\p_a+{1\over 2}(-1)^{p(b)}\p_bX^b=
               \left(
  {1\over 2}\p_b\left(E^{bc}\p_c\p_pE^{pa}\right)
  +(-1)^{p(a)} E^{ab}\p_b U\right)\p_a
   +{1\over 2}\p_b\left(E^{bp}\p_p U\right)\,,
    \end{equation*}
i.e.
     \begin{equation}\label{explicitformula}
\X=X^a\p_a=
               \left(
  {1\over 2}\p_b\left(E^{bc}\p_c\p_pE^{pa}\right)
  +(-1)^{p(a)} E^{ab}\p_b U\right)\p_a\,.
    \end{equation}

\smallskip

  Let us make the following remark regarding the  origin of this
construction, and
in particular about its counterpart for an even Poisson structure.

Let $(M,\P)$ be a usual manifold (without supervariables)
with a Poisson structure defined
by an even Poisson tensor $\P$.  For an arbitrary volume form
$\rh=\rho(x)|Dx|$
on $M$, one can consider the modular vector field  $\X=\X^{(\rh)}$
such that for an arbitrary function $F$
            \begin{equation}\label{batoperator0}
      \X^{(\rh)}(F)=     {\L_{D_F}\rh\over \rh}={\rm div\,}_\rh D_F=
    {1\over \rho(x)}\p_a
    \left(\rho(x)P^{ab}\p_b F\right)\,,
            \end{equation}
where $D_F$ is as usual the Hamiltonian vector field corresponding to
$F$: $D_F G=\{F,G\}_\P$.

One can see that $\X^{(\rh)}$ is a Poisson vector field.
If we choose  another volume form $\rh'=e^G\rh$, then
it is easy to see that   $\X^{(\rh)}$ changes by a Hamiltonian vector field:
  $\X^{(\rh')}=\X^{(\rh)}+D_G$
(compare with equation  \eqref{classisthesame}).  Thus one can assign
to the Poisson manifold $(M,\P)$ an equivalence class $[\X^{(\rh)}]$
of vector field\,\eqref{batoperator0} in the
cohomology group $H^1_{LP}(M,\P)$.
It is how A.Weinstein
defined the modular class of a usual Poisson manifold \cite{W2}.

Weinstein's modular class vanishes
for an even symplectic manifold when the Poisson tensor $\P$ is invertible;
in this case  there exists an invariant
volume form $\rh$.
If $\cal G$ is a Lie algebra with structure constants $c^i_{km}$,
then the modular class of  the Lie-Poisson bracket
  $\{u_i,u_k\}=c_{ik}^mu_m$ is represented by
 the covector $t_m=c^i_{im}$.

The Weinstein construction of modular class
can be easily performed for an even Poisson supermanifold,
but remarkably, for an odd Poisson manifold
equation \eqref{batoperator0} does not define a vector field.
The operation which assigns to an arbitrary function $F$
the divergence of
   the Hamiltonian vector field $D_F$ with respect to an
 {\it odd} Poisson structure is no longer a first
order differential operator. It becomes {\it a second order
 operator}  (see detail in \cite{Khjmp1} and \cite{Kh1})
     \begin{equation}\label{batoperator1}
     \Delta_\rh F={1\over 2}
  {\L_{D_F}\rh\over \rh}=
     {1\over 2}{\rm div\,}_\rh D_F\,.
        \end{equation}
Here $1\over 2$ is just a normalisation coefficient,
 and $D_F G=[F,G]_\E$ where $ [\,,\,]_\E$
is the odd bracket of this odd Poisson manifold.
 We change notation for operation
\eqref{batoperator0} stressing the fact that for
an odd Poisson structure it becomes a second order operator.
 The odd Poisson tensor $\E$ defining the odd Poisson structure
is a symmetric tensor
(see \eqref{supersymmetric})
and it is equal to the principal symbol of
this second order operator $\Delta_\rh$.
(The Poisson tensor $\P$
defining the usual Poisson structure
 and the modular vector field $\X^{(\rh)}$
is an antisymmetric tensor (see detail in \cite{Khjmp1},
\cite{Kh1} and \cite{KhVor1})).

 For example consider
 an odd symplectic case, a non-degenerate $n|n$-dimensional
  odd Poisson manifold. (The equality of
 odd and even dimensions is implied by the non-degeneracy 
of the odd bracket.)
Let $(q^i,\theta_j)$, $i,j=1,\dots,n$ be even and odd Darboux coordinates
of this  supermanifold:
          \begin{equation}\label{darbouxodd}
    [q^i,\theta_j]=\delta^i_j\,,\, [q^i,q^j]=0\,, [\theta_i,\theta_j]=0\,.
          \end{equation}
Then \begin{equation}\label{batoperator1}
  \Delta_\rh F={1\over 2}
      {\L_{D_F}\rh\over \rh}=
       {1\over 2}{\rm div\,}_\rh D_F=
    {\p^2 F(q,\theta)\over \p q^i\p\theta_i}+{1\over 2}[\log\rho,F]\,.
            \end{equation}
 In the case if $\rh=|D(q,\theta)|$ is a coordinate volume form,
then this operator
 becomes the famous Batalin-Vilkovisky $\Delta$-operator
which was introduced in the seminal work\,\cite{BatVyl1}.
(See \cite{Khjmp1}, \cite{SchCMPa}, \cite{Khcmp2}.)

Operator \eqref{batoperator1} and the related
operator on half-densities were studied
 in the article \cite{KhVor1}
for an arbitrary odd Poisson manifold.
In particular, in this article it was noticed that the
 operator $\Delta_\rh^2$  is a vector field and that
$\Delta^2_{\rh'}=\Delta^2_{e^G\rh}=\Delta^2_\rh + D_G$ if $\rh'=e^G\rh$.
 From this, the definition of modular class was
suggested for an odd Poisson manifold. One can
 see that this definition of modular class in \cite{KhVor1} is equivalent
to definition \eqref{definitionofmodularclass} given here.
 On the other hand, the considerations in \cite{KhVor1}
were performed  straightforwardly without
using self-adjointness property of operators on half-densities,
and examples of odd Poisson manifolds with a non-trivial
modular class were not considered.

\begin{remark}\label{noticeit}
  Notice that if $\rh$ is an arbitrary volume form
on an odd Poisson manifold
 $(M,\E)$ then operator $\Delta^{(\rh)}\in \F_\E$ on
half-densities with second order compensation field
  $U=U^{(\rh)}$ constructed with use of  compensation
field $U$ (see equations \eqref{deltaviarho} and
\eqref{potentialforstandardoperator}),
and the operator $\Delta_{\rh}$ on functions constructed
in equation \eqref{batoperator1} are related by equation
          \begin{equation*}\label{noticeit1}
  \Delta^{(\rh)}=\sqrt\rh\circ\Delta_{\rh}\circ{1\over \sqrt \rh}\,.
         \end{equation*}
\end{remark}

\section{Examples. Calculation of modular classes. }

We consider here examples of
second order self-adjoint odd operators which define an
odd Poisson bracket. We study the odd potential of these operators,
 and calculate the modular classes of these Poisson manifolds.

\subsection { An odd symplectic manifold. Bering's formula}

Let $(M,\E)$ be an odd symplectic manifold, i.e. an odd Poisson
tensor field $\E$ defines a non-degenerate odd Poisson structure.
The odd differential form $\e=dx^adx^be_{ab}(x)$ is inverse
to the odd Poisson tensor $\E=E^{ab}(x)\p_b\otimes \p_a$:
$E^{ac}(x)e_{cb}(x)=\delta^a_b$.

 The basic example of an odd symplectic manifold is the following:
for an arbitrary $n$-dimensional manifold $L$, consider its cotangent bundle
reversing the parity of coordinates in the fibres, i.e.
 an $n|n$-dimensional supermanifold $M=\Pi T^*L$. One can assign to
   local coordinates
  $(q^i)$ on the manifold $L$, local coordinates $w^a=(q^i,\theta_j)$,
($i,j=1,\dots,n$) on the supermanifold $\Pi T^*L$,
where $\theta_j$ are odd coordinates in the fibres,
the odd conjugate momenta (under a change of local
 coordinates $(q^i)\mapsto (q^{i'})$,
     $\theta_{j'}={\p q^j\over \p q^{j'}}\theta_j$).
Functions on the supermanifold $M=\Pi T^*L$ can be
identified with multivector
fields on $L$. In the same way as the cotangent bundle
$T^*L$ possesses the usual canonical even symplectic structure,
the supermanifold $\Pi T^*L$  possesses canonical
odd symplectic  structure; local coordinates $(q^i,\theta_j)$
 become Darboux coordinates
(see equation \eqref{darbouxodd})  of this odd symplectic supermanifold.
 The initial manifold  $L$
is a Lagrangian $n|0$-dimensional surface in this symplectic space.
 One can show that every odd symplectic
supermanifold is symplectomorphic to $\Pi T^*L$, where
$L$ is its  $n|0$-dimensional Lagrangian surface
  (see \cite{SchCMPa} and
\cite{Khcmp2}).

   On the odd symplectic supermanifold  $(M,\E)$,
 consider an arbitrary odd second order
self-adjoint operator on half-densities
with principal symbol $\E$. In arbitrary local coordinates
it has the appearance given by equation \eqref{deltainlocalcoordinates}.
An  odd potential $U(x)$ can be defined, for example,
via a volume form or a connection on densities
(see equations \eqref{potentialforstandardoperator}
and \eqref{potentialforstandardoperator2}),
but in the case when $M$ is an odd symplectic supermanifold
one can define an odd
potential $U(x)$ as a primary object in the following way:
   choose instead of local coordinates $(x^a)$, arbitrary
Darboux coordinates  $w^a=(q^i,\theta_j)$ and define the operator
$\Delta$ in these Darboux coordinates by the equation
                  \begin{equation}\label{canonicaloperator}
\Delta\colon\quad  \Delta\bs  =
\Delta\left(s(q,\theta)\sqrt{|D(q,\theta)|}\right)
= {\p^2 s(q,\theta)\over \p q^i\p\theta_i}\sqrt{|D(q,\theta)|}.
     \end{equation}
The remarkable fact is that this formula gives a well-defined operator,
since  expression \eqref{canonicaloperator} does not depend on a choice of
Darboux coordinates  (see for detail \cite{Khcmp2}).
We call this operator the canonical
odd Laplacian  of an odd symplectic manifold.

 In Darboux coordinates, the potential
of the odd canonical Laplacian vanishes.
  Klaus Bering calculated in \cite{Bering1} the
expression for the odd potential $U(x)$
of the canonical odd Laplacian in arbitrary coordinates.
 In our notation
it looks as follows:
        \begin{equation}\label{beringformula}
       U(x)={1\over 4}\p_b\p_a E^{ab}(x)-
        (-1)^{p(b)(p(d)+1)}{1\over 12}\p_a E^{bc}(x) e_{cd}(x)\p_b E^{da}(x).
         \end{equation}
  One can see by straightforward calculations
that the right hand side of this expression transforms
as a second order compensating field
\eqref{transformofcompensationfield}
under infinitesimal and linear coordinate transformations
  (see \cite{Bering1}).
Alternatively one can come to this answer by the following
considerations: if  $(w^a)$
are arbitrary Darboux coordinates, and
$x^a=x^a(w^{a'})$ then
$E^{ab}(x)=(-1)^{p(a)(p(a')+1)}
x^a_{a'}I^{a'b'}x^b_{b'}$,
where $I^{a'b'}$ are components (they are constants) of the matrix of
tensor $\E$ in Darboux coordinates.  Using
\eqref{transformofcompensationfield} we
will arrive at Bering's formula.
  We discuss this formula later in the next section.

The modular vector field of the canonical odd Laplacian
vanishes since $\left(\Delta\right)^2=0$.
This is evident in Darboux coordinates.
Hence the modular class of an odd symplectic manifold vanishes.

  We see that the modular class of an odd symplectic manifold vanishes
          as in the even case,
   despite the absence of an invariant volume form
  (see for detail \cite{KhVor1}).

\subsection {Koszul bracket.}\label{koszul}
Consider in this example another canonical construction
of an odd Poisson structure, which is now not necessarily symplectic.

  Let $(M,\P)$ be an arbitrary usual Poisson manifold. The Poisson structure
 is defined on $M$ by a rank $2$ contravariant 
antisymmetric tensor $\P$
obeying the Jacobi identity.
  One can assign to $(M,\P)$
an odd Poisson supermanifold in the
following way.
Consider as in example\,\ref{superspaceforcanonicalvolumeform}, the
 supermanifold $\Pi TM$ (the tangent bundle to $M$ with
parity reversed fibres).  Functions on $\Pi TM$ can be identified
with  differential forms on the manifold $M$.
It is convenient to work on  $\Pi TM$ in local coordinates
$(x^i,dx^j)$, where $(x^i)$ are coordinates on $M$
 ($x^i$ are even and $dx^i$ odd variables).
     The usual Poisson bracket $\{\,,\,\}=\{\,,\,\}_\P$
on $M$ can be canonically
lifted to an
  odd Poisson bracket (the Koszul bracket) $[\,,\,]$
 on $\Pi TM$ in the following way:
for two arbitrary functions $F,G$ on $M$
         \begin{equation*}
  [F,G]=0\,,\quad [F,dG]=-\{F,G\}_\P\,,\quad [dF,dG]=-d\{F,G\}_\P \,.
          \end{equation*}
These relations give a well-defined odd bracket.
In local coordinates, if $\{x^i,x^j\}=P^{ij}$, then
          \begin{equation}\label{koszul2}
[x^i,x^j]=0,\quad [x^i,dx^j]=[dx^i,x^j]=-P^{ij},\quad
 [dx^i,dx^j]=-dP^{ij}=-dx^k\p_k P^{ij}\,.
          \end{equation}
Thus we define an odd Poisson structure (Koszul bracket)
 $\E_\P$ on $\Pi TM$ via the even Poisson structure on the manifold
 $M$. The matrix of the tensor $\E_\P$ in coordinates $(x^i,dx^j)$
is the following:
             \begin{equation}\label{matrixofE}
    E^{ab}_{\P}(x,dx)=
                    \begin{pmatrix}
                    0 & -P^{ij}(x)\cr
                    P^{ij} & dx^k\p_k P^{ij}(x)\cr
                    \end{pmatrix},
                 \end{equation}
where $||E^{ab}_\P||$ is an $n|n\times n|n$ odd matrix, $n$ is the
dimension of the Poisson manifold $M$, and $||P^{ij}||$
is the $n\times n$ matrix of the Poisson tensor $\P$.

  One can consider a canonical volume form $\rh$ on 
$\Pi TM$ such that in coordinates $(x^i,dx^j)$,
$\rh=|D(x,dx)|$ (see equation \eqref{canonicalvolumeform}).
 Let  $\Delta$ be a second order self-adjoint operator
  on half-densities on the odd Poisson manifold $(\Pi TM,\E_\P)$,
  $\Delta\in \F_{\E_\P}$,
   with a second order compensating field, an
odd potential $U$,  defined by this canonical volume form,
$U=U^{(\rh)},\Delta=\Delta^{(\rh)}$ (see equations
\eqref{deltaviarho} and
\eqref{potentialforstandardoperator}). In other words, the
operator $\Delta$ has principal symbol $\E_\P$
and has the appearance of
\eqref{deltainlocalcoordinates}, where the tensor
$\E_\P=||E^{ab}_\P||$
is defined by the Koszul bracket \eqref{matrixofE},
and the potential  $U=U^{(\rh)}$
vanishes in local coordinates $(x,dx)$.

 The canonical volume form $\rh$ on $\Pi TM$ identifies
 the operator $\Delta=\Delta^{(\rh)}$ on half-densities
 with an operator $\Delta_\rh$ on functions on $\Pi TM$
(see the equation in remark \ref{noticeit}),
and further with the {\it Koszul-Brylinski
operator} $\p_\P$.   The Koszul-Brylinski operator 
         is the operator on 
 differential forms which is defined by the 
equation $\p_\P=[d,{\bf \io}_\P]$, 
where $d=dx^a{\p\over \p x^a}$ is the de Rham differential, and
  $\io_\P$ is the interior product with the bivector $\P$.
  If the volume form is canonical,
the identification of differential forms on $M$ with functions on $\Pi TM$
identifies the two operators $\p_\P$ and $\Delta_\rh$. 
The Jacobi identity for the Poisson tensor $\P$ implies that
$\io_\P^2=0$, and hence for the operator $\Delta$ on
half-densities, $\Delta^2=0$ also.
 Since  $\Delta^2=0$,
  the modular vector field of the operator $\Delta$ vanishes, 
and therefore the modular class of the odd Poisson manifold 
$(\Pi TM, \P)$ vanishes.
(See \cite{yvette} for details, where the relation between the 
Koszul-Brylinski operator and the operator 
$\Delta_\rh$ on functions was studied.)

\begin{remark}
The
vanishing of the modular vector field for operator $\Delta$
on half-densities can be checked
by the following straightforward considerations.
In coordinates $(x^i,dx^j)$, the odd potential $U=U^{(\rh)}$ vanishes,
and the modular vector field $\X$ of the operator $\Delta$
 according to equation \eqref{explicitformula} has the
following appearance
                \begin{equation}\label{modularfieldincoord}
          \X(\Delta)=
          {1\over 2}\p_b\left(
       E^{bc}_\P(x,dx)\p_c\p_pE_\P^{pa}(x,dx)
                 \right)\p_a.
                \end{equation}
Notice that if a point $x$ of manifold $M$ is regular,
 (i.e. rank of Poisson tensor $\P$ is locally
constant in a vicinity of this point,) then one can find
 in a vicinity of this point
Darboux-Lie coordinates such that in these coordinates, the 
 components $||P^{ij}||$ of the Poisson tensor  $\P$
are constants (see \cite{W1} for details).
 This means that the entries of the matrix $||E^{ab}_\P||$ of the odd Poisson
tensor $\E_\P$ (see \eqref{matrixofE}) are
 also constants.
Hence equation \eqref{modularfieldincoord} implies
that $\X(\Delta)$ vanishes at regular points of $M$.  Thus
the modular field vanishes at all points, since regular points are dense
in $M$.

\end{remark}

\subsection {Example of non-trivial modular class}

  Now we consider an example of an
odd Poisson manifold with a non-vanishing
modular class.

  Let $N$ be an arbitrary  $p|q$-dimensional supermanifold.
  We add to $N$ an additional
coordinate --`odd time', by
considering a new $p|q+1$-dimensional supermanifold
       $$
   \N=N\times \Pi \RR\,,
       $$
where $\Pi \RR$ is the $0|1$-dimensional odd line.
Local coordinates on the supermanifold $\N$ are
$(x^a,\tau)$,
where $x^a$ are even and odd coordinates on $N$, and
 $\tau$ is an odd coordinate on $\Pi \RR$.

   The construction of the supermanifold
$\N=N\times\Pi \RR$ is
in a certain precise sense adjoint to the construction of
the supermanifold $\Pi TM$, which we used in
example \ref{superspaceforcanonicalvolumeform} and
in the construction in subsection \ref{koszul}.
Namely
                   $$
\Maps (N\times \Pi \RR, M)=\Maps (N, \Pi TM)\,,\quad
\hbox{(natural isomorphism)}
                   $$
since
                   $$
         \Pi TM=\underline{\Maps}(\Pi\RR,M)\,,\quad
         \hbox{(the supermanifold of `odd paths')}\,.
                   $$

Let  $\et$ be an arbitrary odd vector field on supermanifold $N$,
       $\et=\eta^a(x)\p_a$, ($p(\et)=1$).
   Consider on the supermanifold $\N=N\times \Pi \RR$ two vector fields,
  even and odd, $\A=\tau\et$ and $\B={\p\over \p\tau}$.
These vector fields define a
 second order, odd self-adjoint operator on half-densities
                  \begin{equation}\label{operatordefinesstructure}
\Delta={1\over 2}\left(\L_\A\circ\L_\B+\L_\B\circ\L_\A\right)=
{1\over 2}\left(\L_{\tau\et}\circ\L_{\p_\tau}+
                \L_{\p_\tau}\circ\L_{\tau\eta}\right)\,,
            \end{equation}
where  $\L_\A,\L_\B$  are the Lie derivatives of half-densities.
(See equation \eqref{arbitrarydelta} in
Remark\,\ref{defwithoutvolumeform}.)
In local coordinates
         $$
\L_\A = \tau\eta^a(x)\p_a+{1\over 2}\tau\p \eta,\,
\left(\p \eta=\p_a \eta^a(x)\right)\,,\qquad
\L_\B = {\p\over \p\tau}\,,
         $$
and so
        \begin{equation*}\label{deltainexample3}
 \Delta=\tau\eta^a(x)\p_a\p_\tau+
{1\over 2}\eta^a(x)\p_a+{1\over 2}\tau\p\eta \p_\tau+{1\over 4}\p\eta\,,
        \end{equation*}
One can see by straightforward calculations that
   \begin{equation}\label{modularfieldexample}
\Delta^2={1\over 4}\eta^a\p_a\eta^b\p_b+{1\over 8}\eta^a\p_a\p\eta=
    \L_\X\,,\quad {\rm where}\,\,
    \X=
 \X(\Delta)={1\over 4}\eta^b\p_b\eta^a\p_a=
{1\over 4}\et\et={1\over 8}[\et,\et]\,.
    \end{equation}
  Since $\Delta^2$ is a vector field, Proposition \ref{deltasquare} 
  tells us that the
principal symbol of the operator
$\Delta$ defines an odd Poisson structure $\{\,,\,\}$ on $\N$ 
 given in local coordinates as
                   \begin{equation}\label{oddbracketintheexample}
        \{x^a,\tau\}= -\{\tau,x^a\} =\tau\eta^a\,,
       \{x^a,x^b\}=0\,, \{\tau,\tau\}=0\,.
                    \end{equation}
 Consider the modular vector field $\X$ defined by equation
\eqref{modularfieldexample}. One can see that if
 $\X\not=0$, then the modular class of this Poisson manifold does not vanish.
Indeed, the Poisson bracket \eqref{oddbracketintheexample}
of two arbitrary functions
belongs to the ideal generated by the odd variable $\tau$,
it is proportional to $\tau$, the `odd time'.
Thus all the components of an arbitrary Hamiltonian
vector field are proportional to $\tau$ since
   $D_F=\{F,x^a\}\p_a+\{F,\tau\}\p_\tau$. On the other hand, the
components of the modular vector field \eqref{modularfieldexample}
do not depend on $\tau$. Hence the
modular class $[\X]$ of the modular vector field
\eqref{modularfieldexample} is not trivial in the
cohomology group
$H^1_{LP}(\N,\E)$ of the Poisson manifold $(\N,\E)$ (if $\X\neq 0$).
 The considerations of this example can be summarised in the
following statement.
\begin{theorem}\label{theorem1}
  Let $N$ be an arbitrary supermanifold,
and let $\et$ be an arbitrary odd vector field on this supermanifold.
   The pair $(N,\et)$ defines the odd Poisson manifold
   $\N=N\times \Pi\RR$ with Poisson structure expressed via
  the principal symbol of second order operator
\eqref{operatordefinesstructure} (see equations
\eqref{oddbracketintheexample}).

     The  even vector field  $\X={1\over 8}[\et,\et]$ is the
modular vector field
of this operator ($\Delta^2=\L_\X$),
 and this  vector field (if it does not vanish,)
defines a non-trivial modular
class of this odd Poisson supermanifold.
\end{theorem}

In the next example we consider a
special case of this modular class.

\begin{remark}
     Notice that in particular
   $C(N\times \Pi\RR)=C(N)\oplus C(\Pi N)$.
For an arbitrary function $F(x,\tau)$ on $N\times\Pi \RR$
              $$
F(x,\tau)=f(x)+\tau g(x)\mapsto \begin{pmatrix}f(x)\cr g(x)\cr\end{pmatrix}\,.
                  $$
 One can see that in this representation the actions
\eqref{operatordefinesstructure},
\eqref{modularfieldexample}  of operators
$\Delta$ and $\Delta^2$ have the following appearance:
              $$
\Delta\begin{pmatrix}f(x)\cr g(x)\cr\end{pmatrix}=
{1\over 2}\L_\et\begin{pmatrix}f(x)\cr g(x)\cr\end{pmatrix}\,,
              $$
and       $$
\Delta^2\begin{pmatrix}f(x)\cr g(x)\cr\end{pmatrix}=
{1\over 4}\L^2_\et\begin{pmatrix}f(x)\cr g(x)\cr\end{pmatrix}=
{1\over 4}\L_\et\circ\L_\et
\begin{pmatrix}f(x)\cr g(x)\cr\end{pmatrix}=
{1\over 8}\L_{[\et,\et]}
\begin{pmatrix}f(x)\cr g(x)\cr\end{pmatrix}\,.
              $$
\end{remark}

\subsection {Modular class and Nijenhuis
bracket}

  In this example, using Theorem\,\ref{theorem1}, we study the relation
between the Nijenhuis
bracket of form-valued
vector fields  and modular classes.

   The Nijenhuis
bracket is the bracket of vector fields with
values in differential forms.  Recall its construction.
Let $\X=\X(x,dx)$ be a vector
field  with values in differential forms on a manifold $M$, e.g.
   $\X=dx^jX_j^i(x)\p_i$ is a vector field with values
in $1$-forms on $M$ (which corresponds to
  a linear operator on tangent vectors to $M$).

Consider the supermanifold $\Pi TM$ (see
examples \ref{superspaceforcanonicalvolumeform}, \ref{koszul}).
   Every vector field  $\X(x,dx)$ with values in
differential forms can be lifted  to a vector field
 on the supermanifold $\Pi TM$:
  \begin{equation}\label{lifting}
   \X=X^i(x,dx)\p_i\mapsto \hat\X=X^i(x,dx)\p_i+
           (-1)^{p(\X)}dx^k{\p\over \p x^k}X^i(x,dx)
           {\p\over \p dx^i}.
               \end{equation}
The lifting is uniquely defined by the condition that the lifted
field $\hat\X$
commutes with the de Rham differential $d=dx^k{\p\over \p x^k}$:
             \begin{equation*}
\widehat\X\colon \qquad\begin{cases}
            p_*\widehat \X=\X\,,\quad p\colon \Pi TM\to M\cr
             [\widehat\X,d]=0.\cr
                        \end{cases}
             \end{equation*}
  The Nijenhuis
 bracket $[\X,\Y]_N$ of form-valued vector fields
  $\X,\Y$ can be defined as a form-valued vector field
such that
      \begin{equation*}
               [\X,\Y]_N\colon\quad
\widehat{[\X,\Y]_N}=[\widehat{\X},\widehat{\Y}].
                \end{equation*}
The RHS of this equation is just the usual commutator of vector fields
on the supermanifold $\Pi TM$:
      \begin{equation*}
   [\widehat{\X},\widehat{\Y}]=
      \widehat{\X}\widehat{\Y}-
   (-1)^{p(\hat\X)p(\hat\Y)}\widehat{\Y}\widehat{\X}\,.
                \end{equation*}
(Details of the construction above for the Nijenhuis bracket can be found
  e.g. in \cite{KhVor5}. This bracket was discovered in the 1950s and predates
supermathematics, however it is through the language of supermathematics
that this construction is best described.)

In particular, if $\X(x,dx)$ is odd, i.e. it takes values
in usual differential forms of rank $2k+1$, then
 the bracket $[\X,\X]_N$ of this vector with itself
  is in general, not trivial.
  For example, let $\X=dx^jX_j^i(x)\p_i$ be a vector field with values
in usual $1$-forms on $M$, i.e. $\X=\X(x)$ is a section of the
usual linear operators on the tangent bundle.
  The lifted vector field on $\Pi TM$ is equal to
                       $
       \widehat {\X}(x,dx)=dx^jX_j^i(x)\p_i-dx^k dx^j\p_k X^i_j(x)
        {\p \over \p dx^i}
                       $.
This is an odd vector  field on $\Pi TM$. One can see that
  \begin{equation}\label{bracketoflienaroperatorwithitself}
[\X(x),\X(x)]_N=2dx^jdx^r (X^m_j(x)\p_m X^i_r(x)+\p_r X^m_j(x)
              X_m^i(x))\p_i\,.
          \end{equation}

Now return to modular classes.

   Let $\X(x,dx)$ be an arbitrary odd vector field
  on a   manifold $M$ with values in differential forms,
   e.g. with values in $1$-forms, and
    let $\et=\widehat {\X}(x,dx)$ be its lifting \eqref{lifting}
to the supermanifold $\Pi TM$.
    According to Theorem \ref{theorem1}, the pair $(\Pi TM,\et)$
  defines a supermanifold $\N=\Pi TM\times \Pi \RR$
provided with an induced odd Poisson structure such
  that the modular class of this Poisson manifold
is equal to $$
    {1\over 8}[\et,\et]=
  {1\over 8}[\widehat{\X}(x,dx),\widehat{\X}(x,dx)]=
{1\over 8}\widehat{[\X,\X]_N}\,.
             $$
Given an arbitrary odd form-valued vector field $\X$
on a manifold we define an
odd Poisson supermanifold such that its modular
class is represented  by Nijenhuis bracket $[\X,\X]_N$.


\section{Equation $\Delta^2=0$. Discussion}


           Let $(M,\E)$ be a supermanifold equipped with an odd
  Poisson bracket.
   What can we say about the solutions of the equation
          \begin{equation}\label{nilpotent}
       \Delta^2=0\,,
          \end{equation}
where $\Delta$ is a second order self-adjoint operator on half-densities
 with principal symbol $\E$?  In other words, does there exist an
 odd potential $U(x)$ such that
for the operator
$\Delta={1\over 2}\left(E^{ab}\p_b\p_a+\p_b E^{ba}\p_a+U(x)\right)$,
the condition
$\Delta^2=0$ is obeyed?
It follows from Proposition \ref{deltasquare}
that this  equation has a solution if and only if the modular class
of the odd Poisson manifold $(M,\E)$ vanishes.
 Two operators $\Delta, \Delta'\in\F_\E$ obey the equation
$\Delta^2={\Delta'}^2=0$, where $\Delta'=\Delta+F$,
  if and only if the Hamiltonian vector field $D_F$
vanishes, i.e. $F$ is a Casimir function of the Poisson bracket.

  Equation \eqref{nilpotent} can be rewritten as the condition
of the vanishing of the modular vector field
$\X(\Delta)\colon \Delta^2=\L_\X$.
On setting the modular vector field
(see \eqref{explicitformula})
 equal to zero,
we come to a
system of differential equations in the potential $U$:
          \begin{equation}\label{differentialequations}
U(x)\colon\X(\Delta)=0,\,{\rm i.e.}\,,
    (-1)^{p(a)} E^{ab}\p_b U+
  {1\over 2}\p_b\left(E^{bc}\p_c\p_pE^{pa}\right)
                    =0.
          \end{equation}
One can rephrase the statement above in the following way.
These differential equations have a solution
if and only if the modular class of the Poisson
manifold vanishes, and these solutions are unique
up to a Casimir function $F$ of the Poisson bracket.
  The modular class is the obstruction for a solution
of this equation.

   In the symplectic case (if the Poisson structure is non-degenerate,)
 the potential $U$  is defined up to an an odd constant (an auxiliary
odd parameter).
  For the canonical odd Laplacian
 \eqref{canonicaloperator}, this odd constant is equal to zero.
In the absence of odd constants,
equations \eqref{differentialequations} have unique solutions
which are given by Bering's formula \eqref{beringformula}.

 One can distinguish the following special case of an odd Poisson
manifold when there exists a canonical odd Laplacian.
  Suppose that an odd Poisson manifold consists
      of odd symplectic leaves which are all of the same dimension $p|p$,
 (the tensor $\E$ defining the Poisson structure has constant rank $p|p$),
i.e. in a vicinity of an arbitrary point there are
Darboux-Lie coordinates $(q^i,\theta_j,z^\a)$,  $i,j=1,\dots,p$
($q^i$ are even, $\theta_j$ are odd,) such that
   \begin{equation}
                     \begin{matrix}\label{darbouxlie2}
 \{q^i,\theta_j\}=\delta^i_j,\quad \{q^i,q^j\}=0,\quad
                      \{\theta_i,\theta_j\}=0,\quad (i,j=1,\dots,k)\,, \cr
      \{q^i,z^\a\}=\{\theta_j,z^\a\}=\{z^\a,z^\beta\}=0\,.\cr
                     \end{matrix}
                            \end{equation}
 One can consider the following example of such a Poisson
manifold: let $(M,\P)$ be a usual Poisson manifold of constant rank,
i.e. the Poisson tensor $\P$ has constant rank $2k$ at all points
 of $M$. Then the supermanifold $\Pi TM$ with the Koszul bracket
\eqref{koszul2} consists of $p|p$-dimensional symplectic leaves ($p=2k$).

  In the case when there exists coordinates
satisfying equations \eqref{darbouxlie2},
  one can define a canonical operator in a way similar to
\eqref{canonicaloperator}.
In a vicinity of an arbitrary point, choose Darboux-Lie coordinates
 \eqref{darbouxlie2}, and
 set
$\Delta={\p^2\over \p q^i\p \theta_i}$ in these
local Darboux-Lie coordinates.
  As in equation \eqref{canonicaloperator},
one can see that the operator is well-defined,
  it does not depend on the choice of Darboux-Lie
coordinates.
The potential of this operator vanishes in these coordinates.
 In the article \cite{Bering2} Bering wrote
  the formula for the canonical operator in the Poisson case which covers
this example.

     We must mention the following important fact:
the canonical odd Laplacian
assigns to an arbitrary volume form, the scalar function:
 $\rh\mapsto \s_\rh={\Delta\sqrt\rh\over\sqrt\rh}$. In the article
 \cite{BatBer1}, Batalin and Bering showed that
this function is proportional to the scalar curvature
of a torsion-free affine connection
compatible with the symplectic structure and
this volume form. In the paper \cite{BatBer2} the authors
 generalised this result
  for a class of odd Poisson manifolds
  of constant rank (see \eqref{darbouxlie2}).
The explanation  of this statement
will shed light on the geometry of the compensating field
of the odd canonical Laplacian.

 We will make one final remark.

  It is a standard textbook fact that any Riemannian manifold $(M,{\bf G})$
possesses a unique symmetric connection compatible with the metric
$\G$ (the Levi-Civita connection). An explicit formula
$\G^i_{km}=
{g^{ij}\over 2}\left(\p_kg_{mj}+\p_mg_{kj}-\p_jg_{km}\right)$
expresses the Christoffel symbols of the Levi-Civita connection
in terms of the functions $g_{ik}(x)$,
which are the matrix entries of the Riemannian metric
$\G=g_{ik}(x)dx^i dx^k$.
 Any odd symplectic supermanifold possesses a unique odd potential
$U(x)$ (it is a compensating field of the canonical odd Laplacian
\eqref{canonicaloperator}, which vanishes in Darboux coordinates).
Bering's formula \eqref{beringformula} explicitly expresses this potential
in terms of the functions $E^{ik}(x)$
which are the matrix entries of the non-degenerate Poisson tensor
$\E=E^{ab}\p_b\otimes \p_a$ defining this odd symplectic structure.

For an affine connection compatible with the
 odd symplectic structure, the situation is very different. There
 are many symmetric affine connections
  compatible with a given symplectic  structure on a manifold
 (one may take a connection such that its
  Christoffel symbols vanish in
  given set of Darboux coordinates).
    However there is no formula which expresses
at least one of these connections in terms of the
  functions $\w_{ik}(x)$,
  the matrix entries of
 the closed non-degenerate $2$-form
${\boldsymbol \w}=\w_{ik}dx^i\wedge dx^k$
 defining  symplectic structure.
     Indeed, suppose that there exists an expression
                 $$
    \G^i_{km}=\G^i_{km}\left(\w_{pq},{\p \w_{pq}\over \p x^r},\dots\right)
                 $$
 such that if the functions $\w_{pq}(x)$ are the components
of a closed non-degenerate differential $2$-form, then
the left hand side of this equation defines the Christoffel
symbols of a symmetric affine connection compatible with
the symplectic structure defined by the $2$-form
 $\om=\w_{ik}dx^i\wedge dx^k$.
  In a vicinity of an arbitrary point $x_0$, choose arbitrary
   Darboux coordinates $(x^i)$.
In Darboux coordinates all the functions $w_{ik}(x)$ are constants.
Hence the Christoffel symbols $\G^i_{km}$ are constants in arbitrary
Darboux coordinates.
  This contradicts the existence of non-linear
canonical transformations (transformations from Darboux coordinates to
another Darboux coordinate system).
The absence of an explicit formula comes from
the non-uniqueness.

\section*{Acknowledgments}

  We are happy to acknowledge A.Voronov and B.Kruglikov for encouraging
 us to begin the work on this article. It is a pleasure to also thank
   K.Mackenzie and T.Voronov for very useful
discussions.
   One of us (H.Kh.) learnt much about the peculiarities
  of Poisson geometry through discussions with  A. Bolsinov.
   Many thanks to him.

\end{document}